\documentclass{llncs}

\usepackage{amsmath}
\usepackage{amssymb}
\usepackage{amsfonts}
\usepackage{graphicx,epsfig}
\usepackage{tikz}
\usepackage{array}

\DeclareMathOperator{\Suff}{Suff}
\DeclareMathOperator{\Pref}{Pref}
\DeclareMathOperator{\Fact}{Fact}
\DeclareMathOperator{\MT}{\it{MT}}
\DeclareMathOperator{\alp}{Alph}
\newcommand{\abs}[1]{\lvert #1 \rvert}
\renewcommand{\epsilon}{\varepsilon}
\newcommand{\ord}{\triangleleft}
\newcommand{\rrho}{\blacktriangleleft}

\begin{document}

\title{Universal Lyndon Words}

\sloppy

\date{\today}

\author{Arturo Carpi\inst{1}$^{, \star}$, Gabriele Fici\inst{2} \fnmsep \thanks{Partially supported by Italian MIUR Project PRIN~2010LYA9RH, ``Automi e Linguaggi Formali: Aspetti Matematici e Applicativi''.}, \v{S}t\v{e}p\'an Holub\inst{3}\fnmsep\thanks{Partially supported by the Czech Science Foundation grant
number 13-01832S.}, Jakub Opr\v{s}al\inst{3}$^{, \star\star}$, and Marinella~Sciortino\inst{2}$^{, \star}$}

\authorrunning{A. Carpi et al.}

\institute{
Dipartimento di Matematica e Informatica, Universit\`a di Perugia, Italy\\ \email{carpi@dmi.unipg.it}  \and
Dipartimento di Matematica e Informatica, Universit\`a di Palermo, Italy\\
\email{\{fici,mari\}@math.unipa.it} \and
Department of Algebra, Univerzita Karlova, Czech Republic \\
\email{\{holub,oprsal\}@karlin.mff.cuni.cz}
}

\maketitle

\begin{abstract}
A word $w$ over an alphabet $\Sigma$ is a Lyndon word if there exists an order defined on $\Sigma$ for which $w$ is lexicographically smaller than all of its conjugates (other than itself). We introduce and study \emph{universal Lyndon words}, which are words over an $n$-letter alphabet that have length $n!$ and such that all the conjugates are Lyndon words. 
We show that universal Lyndon words exist for every $n$ and exhibit combinatorial and structural properties of these words. We then define particular prefix codes, which we call Hamiltonian lex-codes, and show that every Hamiltonian lex-code is in bijection with the set of the shortest unrepeated prefixes of the conjugates of a universal Lyndon word. This allows us to give an algorithm for constructing all the universal Lyndon words.
\end{abstract}


\section{Introduction}

A word is called Lyndon if it is lexicographically smaller than all of its conjugate words (other than itself). Lyndon words are an important and well studied object in Combinatorics. Recall, for example, the fact that every Lyndon word is unbordered, or the existence of a unique factorization of any word into a non-decreasing sequence of Lyndon words \cite{lot1}. The definition of Lyndon word implicitly assumes a lexicographic order. Therefore, for different orders, we typically obtain several distinct Lyndon conjugates of the same word. The motivation of this paper is to push the idea to its limits, and ask whether there is a \emph{universal Lyndon word}, that is, a word of length $n!$ over $n$ letters such that for each of its conjugates there exists an order with respect to which this conjugate is Lyndon.

Such a word resembles similar objects known in the literature as universal cycles. A \emph{universal cycle} \cite{Ch92} is a circular word containing every object of a particular type exactly once as a factor. Probably the most prominent example of universal cycles are de Bruijn cycles, which are circular words of length $2^{n}$ containing every binary word of length $n$ exactly once.

The set represented by a universal Lyndon word is the set of all total orders on $n$ letters or, equivalently, all permutations of $n$ letters. The most convenient way is to represent the order $a_1<a_2<\cdots< a_n$ by its ``shorthand encoding'', which is the word $a_1a_2\cdots a_{n-1}$. Jackson \cite{Ja93} showed  that the corresponding universal cycles exist for every $n$ and can be obtained from an Eulerian graph in a manner similar to the generation of de Bruijn cycles. Ruskey and Williams \cite{RuWi10} gave efficient algorithms for constructing shorthand universal cycles for permutations. Our paper can be seen as
a generalization of this concept. Indeed, it is easy to note that every shorthand universal cycle for permutations is a universal Lyndon word (see \cite{shorthand} for more details), but the opposite is not true---that is, there exist universal Lyndon words such that the Lyndon conjugate for some order $a_1<a_2<\cdots< a_n$ does not start with $a_1a_2\cdots a_{n-1}$.

We study the structural properties of universal Lyndon words and give combinatorial characterizations. We then develop a method for generating all the universal Lyndon words. 
This method is based on the notion of Hamiltonian lex-code, which we introduce in this paper.

\section{Notation}

Given a finite non-empty ordered set $\Sigma$ (called the {\it alphabet}), we  let $\Sigma^*$ denote the set of words over the alphabet $\Sigma$. Given a finite word $w =a_1a_2\cdots a_n$, with $n \geq 1$ and $a_i \in \Sigma,$  the length $n$ of $w$ is denoted by $\abs{w}$. The  \textit{empty word} will be denoted by $\varepsilon$ and we set $\abs{\varepsilon}=0.$ We let $\Sigma^{n}$ denote the set of words of length $n$ and by $\Sigma^+$ the set of non-empty words. For $u,v \in \Sigma^+$ we let  $\abs u_v$ denote the number of (possibly overlapping) occurrences  of $v$ in $u.$ For instance, $\abs{011100}_{00}=1$ and $\abs{011100}_{11}=2$.

Given a word $w =a_1a_2\cdots a_{n}$, $a_i\in \Sigma,$ we say  a word $v\in \Sigma^+$ is a  {\it factor} of $w$ if  $v=a_{i}a_{i+1}\cdots a_{j}$ for some integers $i$, $j$ with $1\leq i \leq j \leq n$.
We let  $\Fact(w)$ denote the set of all factors of $w$
and $\alp(w)$ the set of all factors of $w$ of length $1.$ If $i=1$ (resp.,~$j=n$), we say that the factor $v$ is a \emph{prefix} (resp.,~a \emph{suffix}) of $w$. We let $\Pref(w)$  (resp.,~$\Suff(w)$)  denote the set of prefixes (resp.,~suffixes) of the word $w$. The empty word $\epsilon$ is a factor, a prefix and a suffix of any word. A factor (resp.,~a prefix, resp.,~a suffix) of a word $w$ is \emph{proper} if it is different from $\epsilon$ and from $w$ itself.

A \emph{border} of $w$ is a proper prefix of $w$ that is also a suffix of $w$. A word is said to be \emph{unbordered} if it does not have borders.
A word $u$ is a \emph{cyclic factor} of $w$ if $u\in \Fact(ww)$ and $\abs{u}\leq \abs w$. We let $\abs {w}_{u}^{c}$ denote the number of (possibly overlapping) occurrences of $u$ as a cyclic factor of $w$. For instance, $\abs{011100}^{c}_{00}=2$. We say that a word $u$ is \emph{conjugate} to a word $v$ if there exist words $w_{1},w_{2}$ such that $u=w_{1}w_{2}$ and $v=w_{2}w_{1}$. The conjugate is \emph{proper} if both $w_{1}$ and $w_{2}$ are non-empty The conjugacy being an equivalence relation, we can define a \emph{cyclic word} as a conjugacy class of words. Note that $u$ is a cyclic factor of a word $w$ if and only if $u$ is a factor of a conjugate of $w$.

 Every total order on the alphabet $\Sigma$ induces a different lexicographic (or dictionary) order on $\Sigma^{*}$. Recall that the lexicographic order $\ord$ on $\Sigma^{*}$ induced by the order $<$ on the alphabet $\Sigma$ is defined as follows: $u \ord v$ if $u$ is a prefix of $v$ or $za$ is a prefix of $u$ and $zb$ is a prefix of $v$, with $a<b$.
We say that a word $w$ over $\Sigma$ is a \emph{Lyndon word} if there exists a total order on $\Sigma$ such that, with respect to this order $w$ is lexicographically smaller than all of its proper conjugates (or, equivalently, proper suffixes). For example, the word $w=abcabb$ is a Lyndon word, because for the order $a<c<b$ it is the smallest word in its conjugacy class.
Note that a Lyndon word must be primitive (i.e., it cannot be written as a concatenation of two or more copies of a shorter word), and therefore its conjugates are all distinct.

A set of words $X\subset \Sigma^{+}$ is a \emph{code} if for every $x_{1},x_{2},\ldots, x_{h},x'_{1},x'_{2},\ldots, x'_{k}\in X$, if $x_{1}x_{2}\cdots x_{h}=x'_{1}x'_{2}\cdots x'_{k}$, then $h=k$ and $x_{i}=x'_{i}$ for every $1\leq i\leq h$.
For example, $X=\{ab,abb\}$ is a code. Every set $X\subset \Sigma^{+}$ with the property that no word in $X$ is a prefix of another word in $X$ is a code, and is called a \emph{prefix code}.

A \emph{directed graph} (or \emph{digraph}) is a pair $G=(V,E)$, where $V$ is a set, whose elements are called \emph{vertices}, and $E$ is a binary relation on $V$ (i.e., a set of ordered pairs of elements of $V$) whose elements are called \emph{edges}.
The \emph{indegree} (resp.,~\emph{outdegree}) of a vertex $v$ in a digraph $G$ is the number of edges incoming to $v$ (resp.,~outgoing from $v$).
A \emph{walk} in a digraph $G$ is a non-empty alternating sequence $v_{0}e_{0}v_{1}e_{1}\cdots e_{k-1}v_{k}$ of vertices and edges of $G$ such that $e_{i}=(v_{i},v_{i+1})$ for every $i<k$. If $v_{0}=v_{k}$ the walk is \emph{closed}.
A closed walk in a digraph $G$ is an \emph{Eulerian cycle} if it traverses every edge of $G$ exactly once. A digraph is \emph{Eulerian} if it admits an Eulerian cycle. A fundamental property of graphs is that a connected digraph is Eulerian if and only if the indegree of each vertex is equal to its outdegree.
A closed walk in a digraph $G$ is a \emph{Hamiltonian cycle} if it contains every vertex of $G$ exactly once. A digraph is \emph{Hamiltonian} if it admits a Hamiltonian cycle.

In the rest of the paper, we let $\Sigma_{n}$ denote the alphabet $\{1,2,\ldots,n\}$, $n>0$.

\section{Universal Lyndon Words}\label{sec:ulw}

\begin{definition}
A universal Lyndon word (ULW) of degree $n$ is a word over $\Sigma_{n}$ that has  length $n!$ and such that all its conjugates are Lyndon words.
\end{definition}

\begin{remark}
Since there exist $n!$ possible orders on $\Sigma_{n}$, a universal Lyndon word $w$ of degree $n$ has the property that for every order on $\Sigma_{n}$, there is exactly one conjugate of $w$ that is Lyndon with respect to this order; on the other hand, from the definition it follows that a conjugate of a universal Lyndon word cannot be Lyndon for more than one order.
\end{remark}

We consider universal Lyndon words up to rotation, i.e., as cyclic words.

\begin{example}
The only universal Lyndon word of degree $1$ is $1$, and the only universal Lyndon word of degree $2$ is $12$. There are three universal Lyndon words of degree $3$, namely  $212313$, $323121$ and $131232$. Note that these words are pairwise isomorphic (i.e., one can be obtained from another by renaming letters). There are $492$ universal Lyndon words of degree $4$. There are $41$ if we consider them up to isomorphism, and are presented in Tables \ref{tab:jackson4} and \ref{tab:nonjackson4}.
\end{example}

\begin{remark}
It is worth noticing that a universal Lyndon word cannot contain a square (i.e., a concatenation of two copies of the same word) as a cyclic factor. That is, a universal Lyndon word is cyclically square-free. Indeed, if $uu$ is a factor of $w$, then there is a conjugate of $w$ that has $u$ as a border, and it is easily shown that every Lyndon word must be unbordered, and therefore no conjugate of a universal Lyndon word can have a border.
\end{remark}

The following proposition gives a sufficient condition for a word being a ULW.

\begin{proposition}\label{prop:1}
Let $n\geq 2$, and $w$ be a word over $\Sigma_{n}$ such that every permutation of $n-1$ elements of $\Sigma_{n}$ appears as a cyclic factor in $w$ exactly once. Then $w$ is a universal Lyndon word.
\end{proposition}

\begin{proof}
 Suppose that every permutation of $n-1$ elements of $\Sigma_{n}$ appears as a cyclic factor in $w$ exactly once.  Since there are $n!$ such words, this implies that $w$ has length $n!$. Now, for any order $a_{1}<a_{2}< \ldots <a_{n-1}< a_{n}$ over $\Sigma_{n}$, there is exactly one conjugate of $w$ beginning with $a_{1}a_{2}\cdots a_{n-1}$, and this conjugate is Lyndon with respect to this order. So $w$ has exactly $n!$ distinct Lyndon conjugates and therefore is a universal Lyndon word.
\qed \end{proof}

\begin{remark}\label{rem:2}
One might wonder whether it is sufficient to suppose that \emph{each} of $w$'s factors of length ${n-1}$ appears exactly once in the word $w$ to guarantee that $w$ is a ULW. This is not the case. For example, let $n=4$; the word $w=123412431324134214231432$ has $n!$ distinct factors of length $n-1$ but is not a universal Lyndon word, since its conjugate $314321234124313241342142$ is not Lyndon for any order (in fact this is a consequence of the fact that the conjugate $313241342142314321234124$ is Lyndon both for the orders $3<1<2<4$ and  $3<1<4<2$).
\end{remark}

We now use the result of Proposition \ref{prop:1} to show that there exist universal Lyndon words for each degree.

Given an integer $n>2$, the \emph{Jackson graph of degree $n$}, denoted $J(n)$, is a directed graph in which the nodes are the words over $\Sigma_{n}$ that are permutations of $n-2$ letters, and there is an edge from node $u$ to node $v$ if and only if the suffix of length $n-3$ of $u$ is equal to the prefix of length $n-3$ of $v$ and the first letter of $u$ is different from the last letter of $v$. The label of such an edge is set to the first letter of $u$.
In Fig. \ref{fig:Jackson}, the Jackson graph $J(4)$ is depicted.

\begin{figure}
\begin{center}
\[
\begin{tikzpicture}[ver/.style={draw,circle, inner sep=1.5pt}, >={latex},scale=0.8]
\node[ver] (23) at (0,1.4) {$23$};
\node[ver] (41) at (90:4) {$41$};
\node[ver] (14) at (0,-1.4) {$14$};
\node[ver] (32) at (-90:4) {$32$};
\node[ver] (12) at (135:4) {$12$};
\node[ver] (43) at (-135:4) {$43$};
\node[ver] (34) at (45:4) {$34$};
\node[ver] (21) at (-45:4) {$21$};
\node[ver] (31) at (-2,0) {$31$};
\node[ver] (42) at (2,0) {$42$};
\node[ver] (24) at (-4.5,0) {$24$};
\node[ver] (13) at (0:4.5) {$13$};
\foreach \x/\y/\z in {12/24/1,12/23/1,21/13/2,21/14/2,23/31/2,23/34/2,31/12/3,31/14/3,32/24/3,32/21/3,42/21/4,42/23/4,13/32/1,13/34/1,24/41/2,24/43/2,34/41/3,34/42/3,41/13/4,41/12/4,43/31/4,43/32/4,14/42/1,14/43/1}
{\draw[->] (\x)--(\y)node[midway,fill=white,inner sep=1pt]{\z};}
\end{tikzpicture}
\]
\end{center}
\caption{\label{fig:Jackson}The Jackson graph $J(4)$ of degree $4$. Every Eulerian cycle of $J(4)$ is a universal Lyndon word.}
\end{figure}
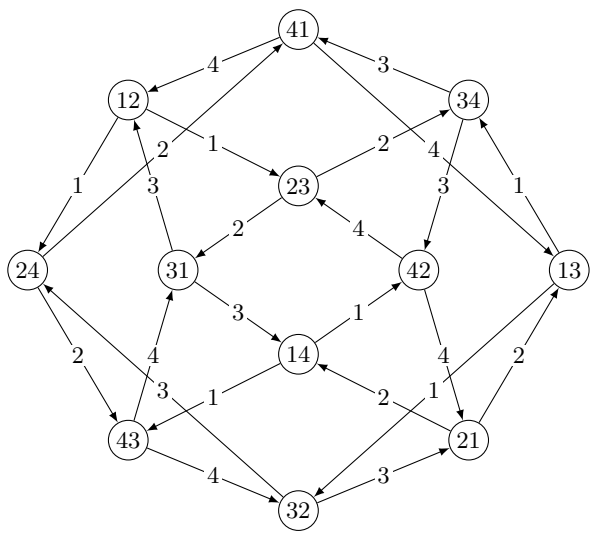

\begin{proposition}
There exist universal Lyndon words of degree $n$ for every $n>0$.
\end{proposition}

\begin{proof}
We can suppose $n>2$. Take the Jackson graph $J(n)$. By construction, this graph is connected and the indegree and outdegree of each vertex are both equal to $2$. Therefore, it contains an Eulerian cycle. Let $w$ denote the word obtained by concatenating the labels of such an Eulerian cycle. Note that every word that is the permutation of $n-2$ letters appears as a cyclic factor in $w$ exactly twice and the two occurrences are followed by the two letters that do not appear in the factor. By Proposition \ref{prop:1}, $w$ is then a universal Lyndon word of degree $n$.
\qed \end{proof}

A universal Lyndon word that is an Eulerian cycle of a Jackson graph is called a \emph{universal cycle} \cite{Ja93}, or \emph{shorthand universal cycle for permutations} \cite{RuWi10}, but in this paper we will call it a \emph{universal Lyndon word of Jackson type}, or simply a \emph{Jackson universal Lyndon word}.

The Jackson universal Lyndon words of degree $4$ are presented in Table \ref{tab:jackson4} (the list contains only pairwise non-isomorphic words, in their representation starting with $1231$).

 \begin{table}[htb]
\begin{center}
  \setlength{\extrarowheight}{2pt}

\begin{tabular}{| >{\centering\arraybackslash}m{5cm} | >{\centering\arraybackslash}m{5cm} |}

    \hline
		& \\[-0.7em]
 $123124132431432142342134$ &    $123124132134214324314234$        \\
 $123124314214321324134234$ &    $123124134213243214314234$        \\
 $123124132431421432134234$  &    $123124314234213214324134$         \\
 $123124314213214324134234$  &    $123124321431423421324134$         \\
 $123124132431423421432134$  &    $123124134213243143214234$       \\
 $123124314214324132134234$  &    $123124132431432134214234$  \\
 $123124132432143142342134$  &    $123124321342143142341324$         \\
 $123124132431432142134234$  &    $123124132143243142134234$          \\
 $123124314234132134214324$  &    $123124132432143142134234$         \\
 $123124314234134214321324$  &    $123124134214321324314234$         \\[-0.7em]
& \\
 \hline
  \end{tabular}\vspace{4mm}
\end{center}
 \caption{The 20 Jackson universal Lyndon words of degree $4$, up to isomorphisms.\label{tab:jackson4}}
\end{table}

However, there are universal Lyndon words that are not of Jackson type. In fact,
the converse of Proposition \ref{prop:1} is not true. For instance, $w=123431242314132421343214$ is a universal Lyndon word of degree 4 but it does not contain any of $142$, $143$, $241$, $243$, $341$, $342$ as a factor.

\section{Order-defining Words}

In this section, we give combinatorial results on the structure of universal Lyndon words.

Let
$w=a_1a_2\cdots a_{n!}$ 
be a universal Lyndon word of degree $n$.
Let $w_i$ denote the conjugate of $w$ starting at position $i$, that is, \[w_i=a_ia_{i+1}\cdots a_{n!}a_1a_2\cdots a_{i-1}\,.\]

\begin{definition}
We say that a partial order $\ord$ on $\Sigma_{n}$ is a \emph{partial alphabet order  with respect to $I\subseteq \Sigma_{n}$} if $\ord$ is a total order on $I$, $i\ord j$ for each $i\in I$ and $j\in \Sigma_{n}\setminus I$, and all $j,k\in \Sigma_{n}\setminus I$ are incomparable. The \emph{size} of $\ord$ is set to $\abs I$.
\end{definition}

Note that a partial alphabet order of size $n-1$ is a total order on $\Sigma_{n}$.

Every word $u\in \Sigma_{n}^+$ defines a partial alphabet order $\ord_u$ with respect to $\alp(u)$, defined as follows: $i\ord_u j$ if and only if the first occurrence of $i$ in $u$ precedes the first occurrence of $j$ in $u$.

The following proposition shows that in a universal Lyndon word, every conjugate is Lyndon with respect to the order it defines. This is an important structural property of universal Lyndon words, which is not true in general. Take, for example, the word $w=123122$. It is Lyndon with respect to the order $1<3<2$, but it is not Lyndon with respect to the order it defines, $1<2<3$.

We let $\ord_i$ denote the order defined by $w_i$ and by $\rrho_i$ the order with respect to which $w_i$ is Lyndon.

\begin{theorem}\label{theor:stepan}
Let $w$ be a word of length $n!$ over $\Sigma_{n}$. Then $w$ is a ULW if and only if every conjugate of $w$ is Lyndon with respect to the order it defines. That is, $\ord_i=\rrho_i$ for every $i$.
\end{theorem}

\begin{proof}
If every conjugate of $w$ is Lyndon, then $w$ is ULW by definition. So we only have to prove the ``only if'' part of the statement.

Suppose that $\ord_j\neq \rrho_j$ for some $j$, and let $k$ be such that $\rrho_k=\ord_j$. Let $z$ be the longest common prefix of $w_j$ and $w_k$. Then $za$ is a prefix of $w_j$ and $zb$ a prefix of $w_k$, where $a\neq b$ are letters. We have $a\rrho_j b$ and $b\rrho_k a$. Therefore, also $b\ord_j a$, which implies that there exists $u\in \Sigma_{n}^{*}$ such that $bua$ is a suffix of $za$ and $bub$ is a suffix of $zb$. Let $w_\ell$ be the conjugate starting with $bua$. Obviously, $b\rrho_\ell a$, since $b$ is the first letter of $w_\ell$. But then we have that $bub\rrho_\ell bua$, and therefore $w_\ell$ has a conjugate smaller than itself for the order $\rrho_\ell$, a contradiction.
\qed \end{proof}

\begin{proposition}\label{minimal2}
Let $w$ be a universal Lyndon word, and $u$ a cyclic factor of $w$. Then for every conjugate $w_{i}$ of $w$, we have that $u$ is a prefix of $w_i$ if and only if $\ord_u\subseteq \rrho_i$.
\end{proposition}

\begin{proof}
By Theorem \ref{theor:stepan}, we have $\ord_u\subseteq \rrho_{i}$ for each $i$ such that $u$ is a prefix of $w_i$. Choose one such $w_i$ (which exists since $u$ is a cyclic factor of $w$). Let $\ord_u\subseteq \rrho_{j}$ and suppose that $za$ is a prefix of $u$ and $zb$ a prefix of $w_j$ for two distinct letters $a$ and $b$ and some $z\in \Sigma_{n}^{*}$. Then $a\rrho_i b$, and, since $a\in \alp(u)$, we deduce that $a\ord_u b$. This implies that $a\rrho_j b$, since $\ord_u\subseteq \rrho_j$. Therefore, $w_i\rrho_j w_j$, a contradiction.
\qed \end{proof}

Proposition \ref{minimal2} states that the cyclic factors of a ULW are in one-to-one correspondence with the orders they define.
As an example, if $1<2$ and, say, $212$ is a cyclic factor of a universal Lyndon word $w$, then every other occurrence of $21$ in $w$ must be followed by $2$.

\begin{corollary}\label{cor:minimal}
Let $w$ be a universal Lyndon word of degree $n$, and  $u$  a cyclic factor of $w$ of length $k>0$. Then $u$ is the lexicographically smallest cyclic factor of $w$ of length $k$ with respect to any total order $\rrho$ on $\Sigma_{n}$ such that $\ord_u\subseteq \rrho$.
\end{corollary}

We now give a combinatorial characterization of universal Lyndon words.

\begin{theorem}\label{theor:number}
 Let $w$ be a word over $\Sigma_{n}$. Then $w$ is a universal Lyndon word if and only if for every cyclic factor $u$ of $w$, one has
\begin{equation}\label{eq:form}
 \abs{w}^{c}_{u}=(n-\abs{\alp(u)})!
\end{equation}
\end{theorem}

\begin{proof}
Suppose that $w$ is a ULW. There are $(n-\abs{\alp(u)})!$ many total orders $\rrho$ on $\Sigma_{n}$ such that $\ord_u\subseteq \rrho$. Hence, (\ref{eq:form}) follows from Corollary \ref{cor:minimal}.

Suppose now that (\ref{eq:form}) holds for every cyclic factor $u$ of $w$ and let us prove that $w$ is a ULW. For every letter $a\in\Sigma_n$, one  has $\abs{w}_a=\abs{w}^c_a=(n-1)!\,$, so that
$\abs{w}=\sum_{a\in\Sigma_n}\abs{w}_a=n!$. Moreover, $w$ is primitive, since $\abs{w}^{c}_{w}=1$.
We show that $w$ is a Lyndon word with respect to $\ord_w$. Let $v$ be a proper conjugate of $w$ and let $z$ be the longest common prefix of $w$ and $v$. Let $a$ and $b$ be the letters that follow the prefix $z$ in $w$ and $v$ respectively.  Since both $za$ and $zb$ occur in $w$, we have $\abs{w}^{c}_{z}>\abs{w}^{c}_{zb}$ which implies $b\notin\alp(z)$ by (\ref{eq:form}). Because $za$ is a prefix of $w$ and $b\notin\alp(z)$, one has $a\ord_w b$, and therefore $w\ord_w v$. This proves that $w$ is a Lyndon word. By a similar argument, all conjugates of $w$ are Lyndon words, so that $w$ is a ULW.
\qed \end{proof}

\begin{corollary}
The reversal of a ULW is a ULW.
\end{corollary}

Note that the fact that the set of universal Lyndon words is closed under reversal is not an immediate consequence of the definition. This property is not true for Lyndon words, e.g. the word $112212$ is Lyndon but its reversal is not. 

\begin{definition}
We say that $u$ is a \emph{minimal order-defining word} if no proper factor of $u$ defines $\ord_{u}$.
\end{definition}

\begin{proposition}
Given a universal Lyndon word $w$ of degree $n$, for each partial alphabet order $\ord$ on $\Sigma_{n}$ there is a unique minimal order-defining word with respect to $\ord$ that is a cyclic factor of $w$.
\end{proposition}

\begin{proof}
Let $\ord$ be a partial alphabet order with respect to $I$.  Let $w_i$ be such that $\ord \subseteq \rrho_i$, and let $u$ be the shortest prefix of $w_i$ such that $\alp(u)=I$.  Note that $\ord_u=\ord$ by Theorem \ref{theor:stepan}.  Clearly, $u$ is a minimal order-defining word, and the uniqueness is a consequence of Proposition \ref{minimal2}.
\qed \end{proof}

Let $w$ be a universal Lyndon word. We let $\MT(w)$ denote the minimal total order-defining words of $w$, i.e., the set of cyclic factors of $w$ that are minimal order-defining words with respect to a total order on $\Sigma_{n}$. The next proposition is a direct consequence of the definitions and of the previous results.

\begin{proposition}\label{prop:mt}
 Let $w$ be a universal Lyndon word of degree $n$, and $u$ a cyclic factor of $w$. The following conditions are equivalent: 
\begin{enumerate}
 \item $u\in \MT(w)$;
 \item $\abs{\alp(u)}=n-1$, and $\abs{\alp(u')}<n-1$ for each proper prefix $u'$ of $u$;
 \item there exists a unique conjugate $w_{i}$ of $w$ such that $u$ is the shortest unrepeated prefix of $w_{i}$.
\end{enumerate}
\end{proposition}

The shortest unrepeated prefix of a word is also called its \emph{initial box} \cite{CaDel01}. 

\medskip

In what follows, we exhibit a structural property of ULW.

\begin{definition}
We say that a cyclic factor $v$ of a word $w$  is a \emph{stretch} if $w$ contains a cyclic factor $avb$ with $a,b\in \Sigma_{n}\setminus \alp(v)$. Let $u$ be a cyclic factor of $w$. We say that a cyclic factor $v$ of $w$ is a \emph{stretch extension} of $u$ in $w$ if $u$ is a factor of $v$, $\alp(u)=\alp(v)$, and $v$ is a stretch.
\end{definition}

Of course, a stretch is always a stretch extension of itself.

\begin{example}
  Let $w=123412431324134214231432$. Then $31$ has two stretch extensions in $w$, namely $313$ and itself.
\end{example}

\begin{lemma}\label{stretch}
Each cyclic factor $u$ of a ULW $w$ has a unique stretch extension in $w$. Moreover, it has a unique occurrence in its stretch extension.
\end{lemma}
\begin{proof}
Let $v$ be a stretch extension of $u$ in $w$. Then $u$ and $v$ have the same number of cyclic occurrences in $w$ by Theorem \ref{theor:number}.
\qed \end{proof}

\begin{theorem}
If $asa$ is a cyclic factor of a ULW $w$, with $a\in \Sigma_{n}\setminus \alp(s)$, then $bsb$ is a cyclic factor of $w$ for each $b\in \Sigma_{n}\setminus \alp(s)$.
\end{theorem}

\begin{proof}
Proceed by induction on $\abs{s}$. The claim trivially holds for $\abs{s}=0$, since $aa$ is not a cyclic factor of $w$. Let now $\abs{s}>0$.
We first show that if $bs$ is a cyclic factor of $w$, then also $bsb$ is a cyclic factor of $w$.  Let therefore $bs$ be a cyclic factor of $w$, where $b\neq a$ is a letter, and let $j$ be such that $\ord_{bsa}\subseteq \rrho_j$.
By Lemma \ref{stretch}, the word $bsa$ 
is not a prefix of $w_j$. Let therefore $bs'e$ be a prefix of $w_j$ and $bs'f$ a prefix of $bsa$ where $e$ and $f$ are distinct letters. Suppose first that $e=b$. If $s'=s$, then $bsb$ is a cyclic factor of $w$ as required. If, on the other hand, the word $s'$ is a proper prefix of $s$, then the induction assumption for the word $bs'b$ implies that $as'a$ is a cyclic factor of $w$. This is a contradiction with Proposition \ref{minimal2} since $\ord_{as'a}\subseteq \ord_{asa}$.
Let now $e\neq b$. Note that then $\ord_{s'e}\subseteq \ord_{sa}$ since $\ord_{bsa}\subseteq \rrho_j$. But we have also $\ord_{s'f}\subseteq \ord_{sa}$, a contradiction with Proposition \ref{minimal2}.

The proof is concluded by a counting argument. Theorem \ref{theor:number} implies that, for any $b\notin\alp(s)$, the word $s$ has $m$ times more cyclic occurrences in $w$ than $bsb$, where $m$ is the cardinality of $\Sigma_{n}\setminus \alp(s)$.
\qed \end{proof}

The previous result shows the combinatorial structure of universal Lyndon words. 
Note that the factors of the form $asa$,  $a\in \Sigma_{n}\setminus \alp(s)$, with $\abs{\alp(s)}< n-2$, only appear in non-Jackson universal Lyndon words. In fact, they can be viewed as premature repetitions of the letter $a$. 

\section{Universal Lyndon Words and Lex-codes}\label{sec:codes}

Proposition \ref{prop:1} implies that an Eulerian cycle in a Jackson graph is a universal Lyndon word. However, there exist universal Lyndon words that do not arise from a Jackson graph,
as we showed at the end of Section \ref{sec:ulw}.

The  non-Jackson universal Lyndon words of degree $4$ are presented in Table \ref{tab:nonjackson4} (the list contains only pairwise non-isomorphic words, in their representation starting with $2123$).

 \begin{table}[h]
\begin{center}
  \setlength{\extrarowheight}{2pt}

\begin{tabular}{| >{\centering\arraybackslash}m{5cm} | >{\centering\arraybackslash}m{5cm} |}

    \hline
		& \\[-0.7em]
 $212313243134212414234143$ &    $212313241432124313414234$        \\
 $212313241423414321243134$ &    $212313241432124313423414$        \\
 $212313421243132414234143$  &    $212313414234212431324143$         \\
 $212313241421243134234143$  &    $212341423132414321243134$         \\
 $212313241423414313421243$  &    $212313212414324313414234$       \\
 $212313243134142342124143$  &    $212313212414324313423414$  \\
 $212313212432414234143134$  &    $212313414234212414313243$         \\
 $212313212414234143243134$  &    $212313212432414313423414$          \\
 $212313212431342341432414$  &    $212313243212414313423414$         \\
 $212313241431342341421243$ &    $212313212432414313414234$         \\
  $212313212431341432414234$ &        \\[-0.7em]
	& \\
 \hline
  \end{tabular}\vspace{4mm}
\end{center}
 \caption{The 21 non-Jackson universal Lyndon words of degree $4$, up to isomorphisms.\label{tab:nonjackson4}}
\end{table}

We now exhibit a method for constructing all the universal Lyndon words. This method is based on particular prefix codes, whose definition is given below.

\begin{definition}
 A set $X\subseteq\Sigma_{n}^*$ is a \emph{lex-code} of degree $n$ if:
 \begin{enumerate}
\item for any $x\in X$, there exists a unique ordering of $\Sigma_n$ such that $x$ is the lexicographical minimum of $X$;
\item if $u$ is a proper prefix of some word of $X$, then $u$ is a prefix of at least two distinct words of $X$.
\end{enumerate}
A lex-code $X$ of degree $n$  is  \emph{Hamiltonian} if the relation
\begin{equation*}\label{eqHG}
S_X=\{(x,y)\in X\times X\mid\exists a\in\Sigma, x \mbox{ is a prefix of } ay\}
\end{equation*}
has a Hamiltonian digraph.
\end{definition}

Notice that Condition 1 in the previous definition ensures that a lex-code is a prefix code.

The following theorem shows the relationships  between Hamiltonian lex-codes and universal Lyndon words. 

\begin{theorem}\label{lexcode}
Let $w$ be a ULW.
Then the set $\MT(w)$
is a Hamiltonian lex-code.
Conversely, if $X\subseteq\Sigma_{n}^*$ is a Hamiltonian lex-code,
then there exists a ULW $w$ such that $X=\MT(w)$.
\end{theorem}

\begin{proof}
We assume that $w$ is a ULW and show that $\MT(w)$ verifies the definition of lex-code. Since there is a bijection between the elements of $\MT(w)$, the conjugates of $w$ (Proposition \ref{prop:mt}) and the total orders on $\Sigma_{n}$ (Theorem \ref{theor:stepan}),  Condition 1 is a direct consequence of Corollary \ref{cor:minimal}.
Always from Proposition \ref{prop:mt}, any proper prefix $x'$ of a word $x$ in $\MT(w)$ contains less than $n-1$ distinct letters. From Theorem \ref{theor:number}, $x'$ has at least two occurrences as a cyclic factor of $w$. Therefore, there exist at least two distinct conjugates $w_{i}$ and $w_{j}$ of $w$ beginning with $x'$. Then $x'$ is a proper prefix of the shortest unrepeated prefixes of $w_{i}$ and $w_{j}$ respectively. By Proposition \ref{prop:mt}, we conclude that Condition 2 holds.

Now, we show that the lex-code $X$ is Hamiltonian. For every $0\leq i\leq n!-1$, let $a_{i}$ be the first letter of the conjugate $w_{i}$ of $w$.
Notice that for every $0\leq i\leq n!-2$ one has
$a_iw_{i+1}=w_ia_i$. By Proposition \ref{prop:mt}, every word in $\MT(w)$ is the shortest unrepeated prefix $x_{i}$ of a conjugate $w_{i}$.
As $x_{i+1}$ is an unrepeated prefix of $w_{i+1}$, the word $v=a_ix_{i+1}$ is an  unrepeated prefix of $a_iw_{i+1}=w_ia_i$.
Thus, either $v=w_ia_i$ or $v$ is an  unrepeated prefix of $w_i$.
In both cases, $x_i$ is a prefix of $v$ and therefore $(x_i,x_{i+1})\in S_X$.
Similarly, one has $(x_{n!-1},x_0)\in S_X$.
We conclude that
$(x_0,x_1,\ldots,x_{n!-1},x_0)$
is a Hamiltonian cycle in the digraph of $S_X$.

Conversely, we assume that $X$ is a Hamiltonian lex-code and show that $X=\MT(w)$ for a suitable ULW  $w$.
Let
$(x_0,x_1,\ldots,x_{k-1},x_0)$ 
be a Hamiltonian cycle in the digraph of $S_X$.
By Condition 1, one has $k=n!$ and $X$ is a prefix code.
Since $(x_i,x_{i+1})\in S_X$, $0\leq i<k$ (where $x_k=x_0$) one has
\begin{equation}\label{eq:SX}
x_iu_i=a_ix_{i+1}
\end{equation}
for suitable $a_i\in\Sigma_n,\ u_i\in\Sigma_n^*,\ 0\leq i<k$.

Set
$
w_i=a_i\cdots a_{k-1}a_0\cdots a_{i-1}$, 
$0\leq i<k$. 
By iterated application of (\ref{eq:SX}), one obtains that $x_i$ is a prefix of a power of ${w_i}$.
Now let $0\leq i,j<k$ and $i\neq j$.
For a sufficiently large $m$, $x_i,x_j$ are prefixes of $w_{i}^m,w_{j}^m$, respectively.
Thus, taking into account that $X$ is a prefix code,
for every total order $\ord$ on $\Sigma_{n}$,
one has ${w_i}\ord{w_j}$ if and only if $x_i\ord x_j$.
From this remark, in view of Condition 1, one derives that $w=w_0$ is a ULW.

To complete the proof, it is sufficient to show that $x_i$ is the shortest unrepeated prefix of $w_{i}$, $0\leq i<k$.
In fact, this implies that $X=\MT(w)$. Suppose that the shortest unrepeated prefix $h_{i}$ of $w_{i}$ is a proper prefix of $x_i$.
Then by Condition 2, $h_{i}$ is also prefix of $x_j$ and consequently of $w_j$, for some $j\neq i$.
But this contradicts Proposition \ref{prop:mt}. Thus $x_{i}$ is a prefix of $h_{i}$.
Now, suppose $x_i\neq h_{i}$. Since by Proposition \ref{prop:mt}, $h_{i}$ is a shortest word containing $n-1$ distinct letters, $\abs{\alp(x_{i})}<n-1$ and, by Theorem \ref{theor:number}, $x_{i}$ has at least another occurrence starting at a position $j\neq i$. So we have that the words $x_{i}$ and $x_{j}$ are one a prefix of the other, against the fact that $X$ is a prefix code.  
\qed
\end{proof}

From Theorem \ref{lexcode}, in order to produce a ULW, one can construct a lex-code and check whether it is Hamiltonian. Let $S_n$ be the set of the total orders on $\Sigma_n$.
All lex-codes of degree $n$  can be obtained by a construction based on iterated refinements of a partition
of $S_n$ as follows:

\begin{enumerate}
\item set $X=\{\epsilon\}$ and $C_\epsilon=S_n$;
\item repeat the following steps until $C_x$ is a singleton for all $x\in X$:
	\begin{enumerate}
	\item select $x\in X$ such that $C_x$ contains at least two elements;
	\item \label{stepx} choose $\Gamma\subseteq\Sigma_n$;
	\item for any $a\in\Gamma$, let $C_{xa}$ be the set of the orders of $C_x$ such that $a=\min\Gamma$;
	\item replace $X$ by $(X\setminus\{x\})\cup\{xa\mid a\in\Gamma,\ C_{xa}\neq\emptyset\}$.
	\end{enumerate}
\end{enumerate}

An example of execution of the previous algorithm is presented in Ex. \ref{ex:lexcode}.

One can verify that after each iteration of loop 2,
 $X$ is a prefix code,
 $(C_x)_{x\in X}$ is a partition of $S_n$,
 and any $x\in X$ is the lexicographic minimum of $X$ for all orders of $C_x$.
It follows that the procedure halts when $X$ is a lex-code.
Moreover, one can prove that any lex-code $X$ may be obtained by the procedure above,
choosing conveniently $\Gamma$ at step (b) of each iteration.

Clearly, not all lex-codes are Hamiltonian.
Thus, the main problem is to understand which limitations the Hamiltonianicity of the lex-code imposes to the construction above.
For example, the words in a lex-code can be arbitrarily long. But by Theorem \ref{lexcode}, if $X$ is a lex-code of degree $n$ and $u\in X$ is longer than $n!$, then $X$ cannot be Hamiltonian.

\begin{example}\label{ex:lexcode}
Let $n=3$. At the beginning of the algorithm, $X=\{\epsilon\}$ and $C_{\epsilon} = S_{3} = \{123, 132, 213, 231, 312, 321\}$. The first choice of a word $x$ in $X$ is forced, we must take $x=\epsilon$. Let us choose $\Gamma=\{1,2\}$. We then get $C_{1}=\{123,132,312\}$, $C_{2}=\{213,231,321\}$ and  $X$ becomes $\{1,2\}$. Let us now choose $x=1$ and $\Gamma=\{1,3\}$. We get $C_{11}=\{123,132\}$, $C_{13}=\{312\}$ and therefore $X=\{2,11,13\}$. Next, take $x=2$ and $\Gamma=\{2,3\}$; now $C_{22}=\{213,231\}$, $C_{23}=\{321\}$ and $X=\{11,13,22,23\}$. Then pick $x=11$ and $\Gamma=\{2,3\}$, so that $C_{112}=\{123\}$, $C_{113}=\{132\}$ and $X=\{13,22,23,112,113\}$. Finally, the last choice of a word in $X$ is forced, $x=22$ (since $C_{22}$ is the only set of cardinality greater than $1$ left). We choose $\Gamma=\{1,3\}$ and get $C_{221}=\{213\}$ and $C_{223}=\{231\}$. The lex-code obtained is thus $X=\{13,23,112,113,221,223\}$. The reader can verify that this lex-code is not Hamiltonian.

The following choices of $x$ and $\Gamma$ lead to the Hamiltonian lex-code $X=\{12,13,21,23,31,32\}$: $x=\epsilon$, $\Gamma=\{1,2,3\}$; $x=1$, $\Gamma=\{2,3\}$; $x=2$, $\Gamma=\{1,3\}$; $x=3$, $\Gamma=\{1,2\}$.
\end{example}

\section{Conclusion and Open Problems}

We introduced universal Lyndon words, which are words over an $n$-letter alphabet having $n!$ Lyndon conjugates. We showed that this class of words properly contains the class of  shorthand universal cycles for permutations. We gave combinatorial characterizations and constructions for universal Lyndon words. We leave open the problem of finding an explicit formula for the number of ULW of a given degree.

We exhibited an algorithm for constructing all the universal Lyndon words of a given degree. The algorithm is based on the search for a Hamiltonian cycle in a digraph defined by a particular code, called Hamiltonian lex-code, that we introduced in this paper. It would be natural to find efficient algorithms for generating (or even only counting) universal Lyndon words. 

Finally, universal Lyndon words have the property that every conjugate defines a different order, with respect to which it is Lyndon. We can define a \emph{universal order word} as a word of length $n!$ over $\Sigma_{n}$ such that every conjugate defines a different order. Universal Lyndon words are therefore universal order words, but the converse is not true, e.g. the word $123421323121424314324134$ is a universal order word but is not ULW. Thus, it would be interesting to investigate which properties of universal Lyndon words still hold for this more general class.


\end{document}